\documentclass[11pt]{amsart}

\usepackage{algorithm}
\usepackage{algpseudocode}
\usepackage{epigraph}
\usepackage{mathtools}
\usepackage{latexsym}
\usepackage{mathrsfs}
\usepackage{pstricks}
\usepackage{listings}
\usepackage{hyperref}
\usepackage{pgfplots}
\usepackage{appendix}
\usepackage{pst-all}
\usepackage{pstricks,pst-plot}
\pgfplotsset{compat=1.18} 
\usepackage{mathdots}
\usepackage{xcolor}
\usepackage{enumitem}
\usepackage{todonotes}
\usepackage{amssymb,bm}
\usepackage{amsmath}
\usepackage{amsfonts}
\usepackage{amsthm}
\usepackage{tikz}
\usepackage{tikz-cd}
\usepackage{epsfig}
\usepackage{verbatim}
\usepackage{float}
\usepackage{tabularx}
\usepackage{bm}
\usepackage{mathtools}
\usepackage{blkarray, bigstrut}
\usepackage{subcaption}
\usepackage{algorithm}
\usepackage{algpseudocode}

\usepackage{graphicx}
\usepackage[margin=1in]{geometry}
\usepackage{fmtcount}
\usepackage{comment}
\usepackage{setspace}
\usepackage{bbold}

\definecolor{codegray}{rgb}{0.5,0.5,0.5}
\lstdefinestyle{mystyle}{
    commentstyle=\color{codegreen},
    keywordstyle=\color{magenta},
    numberstyle=\tiny\color{codegray},
    stringstyle=\color{codepurple},
    basicstyle=\ttfamily\footnotesize,
    breakatwhitespace=false,         
    breaklines=true,                 
    captionpos=b,                    
    keepspaces=true,                 
    numbers=left,                    
    numbersep=5pt,                  
    showspaces=false,                
    showstringspaces=false,
    showtabs=false,                  
    tabsize=2
}
  \lstdefinelanguage{GAP}{
    basicstyle=\ttfamily,
    keywords={true, false, function, return, fail, if, in, while, do, od, else, elif, fi, break, continue},
    keywordstyle=\color{blue}\bfseries,
    otherkeywords={
      >, <, ==
    },
    breaklines=true,      
    identifierstyle=\color{black},
    sensitive=True,
    comment=[l]{\#},
    commentstyle=\color{cyan},
    stringstyle=\color{red},
    morestring=[b]',
    morestring=[b]"
  }

\providecommand{\U}[1]{\protect\rule{.1in}{.1in}}

\newcolumntype{Y}{>{\raggedleft\arraybackslash}X}
\def\bc{{\mathbb{C}}}

\def\bn{{\mathbb{N}}}

\def\br{{\mathbb{R}}}

\def\bz{{\mathbb{Z}}}

\def\br{\mathbb R}

\def\wt{\widetilde}

\def\vs{\vskip.3cm}

\def\wt{\widetilde}
\def\gdeg{G\text{\rm -deg}}

\def\t2deg{\mathbb T^2\text{\rm -deg}}

\def\s1deg{S^1\text{\rm -deg}}

\def\Om{\Omega}

\DeclareMathOperator{\id}{Id}

\newrgbcolor{violet}{.6 .1 .8}
  \newrgbcolor{lightyellow}{1 1 .8}
  \newrgbcolor{lightblue}{.80 1 1}
  \newrgbcolor{mygreen}{0 .66 .05}
  \definecolor{mygreen}{rgb}{0,.66,.05}
  \definecolor{lightyellow}{rgb}{1,1,.80}
  \newrgbcolor{orange}{1 .6 0}
  \newrgbcolor{GreenYellow}{.85 1 .31}
  \newrgbcolor{Yellow}{1  1  0}
  \newrgbcolor{Goldenrod}{1  .90  .16}
  \newrgbcolor{Dandelion}{1  .71  .16}
  \newrgbcolor{Apricot}{1  .68  .48}
  \newrgbcolor{Peach}{1  .50  .30}
  \newrgbcolor{Melon}{1  .54  .50}
  \newrgbcolor{YellowOrange}{1  .58  0}
  \newrgbcolor{Orange}{1  .39  .13}
  \newrgbcolor{BurntOrange}{1  .49  0}
  \newrgbcolor{Bittersweet}{1.  .4300  .24}
  \newrgbcolor{RedOrange}{1  .23  .13}
  \newrgbcolor{Mahogany}{1.  .4475  .4345}
  \newrgbcolor{Maroon}{1.  .4084  .5376}
  \newrgbcolor{BrickRed}{1.  .3592  .3232}
  \newrgbcolor{Red}{1  0  0}
  \newrgbcolor{OrangeRed}{1  0  .50}
  \newrgbcolor{RubineRed}{1  0  .87}
  \newrgbcolor{WildStrawberry}{1  .04  .61}
  \newrgbcolor{CarnationPink}{1  .37  1}
  \newrgbcolor{Salmon}{1  .47  .62}
  \newrgbcolor{Magenta}{1  0  1}
  \newrgbcolor{VioletRed}{1  .19  1}
  \newrgbcolor{Rhodamine}{1  .18  1}
  \newrgbcolor{Mulberry}{.6668  .1180  1.}
  \newrgbcolor{RedViolet}{.9538  .4060  1.}
  \newrgbcolor{Fuchsia}{.5676  .1628  1.}
  \newrgbcolor{Lavender}{1  .52  1}
  \newrgbcolor{Thistle}{.88  .41  1}
  \newrgbcolor{Orchid}{.68  .36  1}
  \newrgbcolor{DarkOrchid}{.60  .20  .80}
  \newrgbcolor{Purple}{.55  .14  1}
  \newrgbcolor{Plum}{.50  0  1}
  \newrgbcolor{Violet}{.98 .15 .95}
  \newrgbcolor{RoyalPurple}{.25  .10  1}
  \newrgbcolor{BlueViolet}{.84  .38  .98}
  \newrgbcolor{Periwinkle}{.43  .45  1}
  \newrgbcolor{CadetBlue}{.38  .43  .77}
  \newrgbcolor{CornflowerBlue}{.35  .87  1}
  \newrgbcolor{MidnightBlue}{.4414  .9259  1.}
  \newrgbcolor{NavyBlue}{.06  .46  1}
  \newrgbcolor{RoyalBlue}{0  .50  1}
  \newrgbcolor{Blue}{0  0  1}
  \newrgbcolor{Cerulean}{.06  .89  1}
  \newrgbcolor{Cyan}{0  1  1}
  \newrgbcolor{ProcessBlue}{.04  1  1}
  \newrgbcolor{SkyBlue}{.38  1  .88}
  \newrgbcolor{Turquoise}{.15  1  .80}
  \newrgbcolor{TealBlue}{.1572  1.  .6668}
  \newrgbcolor{Aquamarine}{.18  1  .70}
  \newrgbcolor{BlueGreen}{.15  1  .67}
  \newrgbcolor{Emerald}{0  1  .50}
  \newrgbcolor{JungleGreen}{.01  1  .48}
  \newrgbcolor{SeaGreen}{.31  1  .50}
  \newrgbcolor{Green}{0  1  0}
  \newrgbcolor{ForestGreen}{.1992  1.  .2256}
  \newrgbcolor{PineGreen}{.3100  1.  .5575}
  \newrgbcolor{LimeGreen}{.50  1  0}
  \newrgbcolor{YellowGreen}{.56  1  .26}
  \newrgbcolor{SpringGreen}{.74  1  .24}
  \newrgbcolor{OliveGreen}{.6160  1.  .4300}
  \newrgbcolor{RawSienna}{.53  .28  .16}
  \newrgbcolor{Sepia}{1.  .7510  .70}
  \newrgbcolor{Brown}{.41  .25  .18}
  \newrgbcolor{TAN}{.86  .58  .44}
  \newrgbcolor{Gray}{1.  1.  1.}
  \newrgbcolor{Black}{1  1  1}
  \newrgbcolor{White}{1  1  1}

\newtheorem{theorem}{Theorem}[section]
\newtheorem{proposition}{Proposition}[section]

\newtheorem{corollary}{Corollary}[section]

\newtheorem{remark}{Remark}[section]

\newtheorem{remark-definition}{Remark and Definition}[section]
\newtheorem{rem-not}{Remark and Notation}[section]

\title[A Burnside Ring Cryptosystem]{A Symmetric-Key Cryptosystem Based on the Burnside Ring of a Compact Lie Group}
\author{Ziad Ghanem}
\date{}

\begin{document}

\maketitle

\begin{abstract}
Classical linear ciphers, such as the Hill cipher, operate on fixed, finite-dimensional modules and are therefore vulnerable to straightforward known-plaintext attacks that recover the key as a fully determined linear operator. We propose a symmetric-key cryptosystem whose linear action takes place instead in the Burnside ring $A(G)$ of a compact Lie group $G$, with emphasis on the case $G=O(2)$. The secret key consists of (i) a compact Lie group $G$; (ii) a secret total ordering of the subgroup orbit-basis of $A(G)$; and (iii) a finite set $S$ of indices of irreducible $G$-representations, whose associated basic degrees define an involutory multiplier $k\in A(G)$. Messages of arbitrary finite length are encoded as finitely supported elements of $A(G)$ and encrypted via the Burnside product with $k$. For $G=O(2)$ we prove that encryption preserves plaintext support among the generators $\{(D_1),\dots,(D_L),(SO(2)),(O(2))\}$, avoiding ciphertext expansion and security leakage. We then analyze security in passive models, showing that any finite set of observations constrains the action only on a finite-rank submodule $W_L\subset A(O(2))$, and we show information-theoretic non-identifiability of the key from such data. Finally, we prove the scheme is \emph{not} IND-CPA secure, by presenting a one-query chosen-plaintext distinguisher based on dihedral probes.
\end{abstract}

\vs

\setlength{\epigraphwidth}{0.8\textwidth}
\epigraph{If he had anything confidential to say, he wrote it in cipher, that is, by so changing the order of the letters of the alphabet, that not a word could be made out.}{Suetonius, The Lives of the Twelve Caesars \cite{Suetonius}}

\section{Introduction}
A central goal of symmetric cryptography is to transform plaintexts (unencrypted messages) into ciphertexts (encrypted messages) that remain computationally unintelligible to an adversary, while enabling fast, reversible decoding by legitimate parties with access to a shared, secret key. Classical linear schemes like the Hill cipher operate on fixed-length blocks using an invertible matrix over $\mathbb{Z}_{26}$. As is well known, $m$ linearly independent plaintext--ciphertext pairs for block size $m$ suffice to recover the key via linear algebra. This vulnerability stems from the fact that the secret key lives as a fully determined linear operator on a finite-dimensional space. Modern cryptography has formalized strong security notions to address such vulnerabilities. Semantic security, the gold standard, requires that a ciphertext reveal no useful information about its plaintext beyond what could be inferred without observing the ciphertext. This principle is made rigorous through complexity-theoretic frameworks such as Indistinguishability under Chosen-Plaintext Attack (IND-CPA) \cite{BFO08}, which demands that an adversary with access to chosen-plaintext encryptions cannot distinguish between ciphertexts of two messages of their choice.
\vs
The looming threat of quantum computing has intensified the search for cryptographic primitives beyond traditional number-theoretic assumptions. This search for post-quantum primitives has led to schemes based on group theory and graph theory. For example, the Permutation Group Mappings (PGM) cryptosystem introduced by Magliveras and Memon leverages the combinatorial complexity of factorization in finite permutation groups rather than traditional number-theoretic assumptions \cite{Magliveras92,Magliveras89}. Similarly, graph-theoretic approaches have emerged, such as symmetric block ciphers based on the structural invariants of Minimum Spanning Trees (MST) and adjacency matrices of finite weighted graphs \cite{Murthy}. While these algebraic systems provide important alternatives, they share a foundational property with the classical linear ciphers they aim to replace: their cryptographic action takes place in a finite-dimensional platform. This creates a conceptual vulnerability, as a sufficient number of plaintext--ciphertext pairs can fully constrain the secret key operator, enabling its recovery through linear or combinatorial reconstruction.

Our work addresses this limitation by moving beyond finite-dimensional settings. We propose a symmetric-key cryptosystem whose encryption operates in the Burnside ring $A(G)$ of a compact Lie group $G$, an infinite-rank module. 
Our contributions are threefold: First, we give a protocol that encodes arbitrary finite data as finitely supported elements of $A(G)$ and encrypts by Burnside multiplication with $k$; for $G=O(2)$ we prove that encryption does not increase the maximal dihedral index present in the plaintext, avoiding ciphertext expansion and preserving support. Second, we analyze passive security (ciphertext-only and known-plaintext). Observations are confined to a finite-rank submodule $W_L\subset A(O(2))$ determined by the observed supports, while distinct key-sets $S$ can induce \emph{identical} linear operators on $W_L$; we formalize this as information-theoretic non-identifiability from any finite data. Third, we show that our deterministic linear construction cannot meet IND-CPA: we exhibit a one-query distinguishing attack via dihedral probes; this aligns with the general impossibility of deterministic encryption achieving standard indistinguishability (see Bellare–Boldyreva–O’Neill and follow-ups \cite{BBO07}). Although our scheme is not intended to compete with modern high-assurance primitives, we view it as a case study at the interface of algebraic topology, representation theory, and cryptography.

\section{The Burnside Ring}\label{sec:burnside}
\subsection{Equivariant Notation}
Let $G$ be a compact Lie group. For any subgroup  $H \leq G$, we denote by $(H)$ its conjugacy class,
by $N(H)$ its normalizer annd by $W(H):=N(H)/H$ its Weyl group in $G$. Moreover, for any pair of subgroups $(H,K) \in G \times G$, we denote by $n(H,K)$ the number of subgroups $\wt K \leq G$ with $\wt K \in (K)$ and $H \leq \tilde K$.
The set of all subgroup conjugacy classes $G$, denoted by $\Phi(G):=\{(H): H\le G\}$, admits the following natural partial ordering
\[
(H)\leq (K) \iff \exists_{ g\in G}\;\;gHg^{-1}\leq K.
\]
As is possible with any partially ordered set, we extend the natural order over $\Phi(G)$ to a total order, which we indicate by `$\prec$' to differentiate the two relations. According to our choice in the total ordering of $\Phi(G)$, the subset of subgroup conjugacy classes in $G$ admitting a zero-dimensional Weyl group
$\Phi_0 (G):= \{ (H) \in \Phi(G) : \text{$W(H)$  is finite}\}$ can be enumerated $\Phi_0 (G) = \{ (H_1), (H_2), \ldots \}$ in such a way that, if $i < j$, then $(H_i) \prec (H_j)$. 
\vs
The free $\mathbb{Z}$-module $A(G) := \mathbb{Z}[\Phi_0(G)]$ is an abelian group. Since every element $a \in A(G)$ is a \emph{finite} formal sum of the form
\[
a = \sum_{(H) \in I} n_H (H), \quad n_H \in \bz, \; I \subset \Phi_0(G), \; |I| < \infty,
\]
there exists a well-defined isomorphism $\rho: A(G) \rightarrow \bz^{(\bn)}$ associating each $a \in A(G)$ with its vector of integer coefficients, where $\bz^{(\bn)}$ is the set of infinite length integer vectors with compact support. Moreover, we can identify the coefficient standing next to the generator $(H)$ in any element $a \in A(G)$ using the notation
\[
\operatorname{coeff}^H(a) := n_H.
\]
We define the Burnside product of any pair of generators $(H),(K) \in \Phi_0(G)$ as follows
\begin{align*} 
    (H) \cdot (K) := \sum\limits_{(L) \in \Phi_0(G)} n_L(L), 
\end{align*}
where the coefficients $n_L \in \mathbb{Z}$, counting the number of orbits of type $(L)$ in the $G$-set $G/H \times G/K$, are given by the recurrence formula
\begin{align*} 
    n_L := \frac{n(L,H) |W(H)| n(L,K) |W(K)| - \sum_{(\wt L) \prec (L)} n_{\wt L} n(L,\wt L) |W(\wt L)|}{|W(L)|}.
\end{align*}
Extending this product bilinearly to all elements of $A(G)$, we obtain the \emph{Burnside Ring} of $G$.

\subsection{The Burnside Ring $A(O(2))$}
For the remainder of this paper, we focus on the case $G=O(2)$. The set $\Phi_0(O(2))$ is countably infinite, consisting of the conjugacy class of the group itself, $(O(2))$; the class of the special orthogonal group, $(SO(2))$; and the classes of the finite dihedral subgroups, $(D_k)$ for $k \ge 1$. Following \cite{AED}, we characterize the multiplicative structure of $A(O(2))$ by describing the Burnside product rules for the basis generators of $A(O(2))$ in Table \ref{table:multiplication_table}.
\begin{center} \label{table:multiplication_table}
\begin{tabular}{|c|c|c|c|}
\hline
$\cdot$ & $(O(2))$ & $(SO(2))$ & $(D_m)$ \\ \hline
$(O(2))$ & $(O(2))$ & $(SO(2))$ & $(D_m)$ \\ \hline
$(SO(2))$ & $(SO(2))$ & $2(SO(2))$ & $0$ \\ \hline
$(D_k)$ & $(D_k)$ & $0$ & $2(D_{l}), \; l := \gcd(k,m)$ \\ \hline
\end{tabular}
\end{center}
\subsection{Properties of the Basic Degrees in $A(O(2))$}
In order to define the basic degrees in $A(O(2))$, we must first describe the irreducible representations of $O(2)$. These consist of the trivial representation $\mathcal V_0 \simeq \br$ and for each $m \in \bn$ the irreducible representation $\mathcal V_m \simeq \bc$ equipped with the following $m$-folded $O(2)$-action
\[
e^{i\theta} v := e^{i m \theta} \cdot v, \; \kappa v := \overline{v}, \quad v \in \mathcal V_m,
\]
where `$\cdot$' indicates the standard complex multiplication and `$\overline{u}$' the complex conjugation of $u \in \bc$.

Leveraging the formula \eqref{eq:RF-1} and knowledge of the dimensions of the fixed point spaces $\mathcal V_k^H$ for all $k = 0,1,\ldots$ and $H \leq G$, one can easily demonstrate that the trivial irreducible $G$-representation $\mathcal V_0$ admits the basic degree $\deg_{\mathcal V_0} = O(2)$ and, for each $m \in \bn$, one has
\[
\deg_{\mathcal V_m} = (O(2)) - (D_m).
\]
Leveraging the multiplication rules summarized in Table \ref{table:multiplication_table}, let's consider the Burnside product of any two basic degrees $\deg_{\mathcal V_m}$ and $\deg_{\mathcal V_k}$ associated with non-trivial irreducible representations $\mathcal V_m$ and $\mathcal V_k$:
\begin{align*}
    \deg_{\mathcal V_m} \cdot \deg_{\mathcal V_k} &= \left( (O(2)) - (D_m) \right) \cdot \left( (O(2)) - (D_k) \right) \\
    &= O(2) - (D_m) - (D_k) - (D_m) \cdot (D_k) \\
    &= (O(2)) - (D_m) - (D_k) - 2 (D_l), \quad l := \gcd(m,k).
\end{align*}
Clearly, in the case that $m = k$, one has $\deg_{\mathcal V_m}^289 = (O(2))$. Of course, one might be interested in the Burnside product of any finite number of basic degrees. The following result generalizes the above considerations to this setting, where for the sake of simplifying our exposition, we have employed the {\it Iverson brackets} which map any logical predicate $P$ to the set $\{0,1\}$ according to the rule
\begin{align} \label{notation_iverson}
    [P] := \begin{cases}
        1 \quad & \text{ if } P \text{ is true}; \\
        0 \quad & \text{ otherwise}.
    \end{cases}
\end{align}
\begin{proposition}\label{prop:key-coeff}
    For any finite number of basic degrees $\{ \deg_{\mathcal V_{s_j}} \}_{j = 1}^N$ where $s_1,\ldots,s_N \in \bn$ are distinct, and for any number $s_0 \in \bn$, one has
\begin{align*}
\operatorname{coeff}^{D_{s_0}} \left( \prod_{j=1}^N  \deg_{\mathcal V_{k_j}} \right)  = - [s_0 \in \{ s_1,\ldots,s_N \} ] + 2 \sum\limits_{\substack{ I \subset \{s_1,\ldots,s_N\} \\ I \neq \emptyset, \{s_0\} }} (-2)^{\vert I \vert-2} 
 \left[s_0 = \operatorname{gcd}(I)  \right].
\end{align*}
\end{proposition}
\begin{proof}
Consider the Burnside ring product of the relevant basic degrees
\begin{align} 
\label{prod_1}
    \prod\limits_{k=1}^N \deg_{\mathcal V_{s_k}} &= \prod\limits_{k=1}^N  (O(2)) - (D_{s_k}),
\end{align}
which admits the expansion 
\begin{align*} 
    \prod\limits_{k=1}^N  (O(2)) - (D_{s_k}) &  = \sum\limits_{I \subset \{s_1,\ldots,s_N\}} \prod\limits_{s \in I} -(D_s) \nonumber \\
    &= \sum\limits_{I \subset \{s_1,\ldots,s_N\}}(2)^{\vert I \vert -1} (-1)^{|I|}
    (D_{\operatorname{gcd}(I)}),
\end{align*}
where the notation $I$ is used to indicate a subset of the indices $\{s_1,s_2,\ldots, s_N \}$ and the expression $\sum_{I \subset \{s_1,s_2,\ldots, s_N \}}$ describes a summation over all such subsets, including the empty set, in which case we put $\prod_{s \in \emptyset} - (D_s) := (O(2))$, and the full set. It follows that the coefficient of $(D_{s_0}) \in \Phi_0(G)$ in the Burnside product \eqref{prod_1}, which for convenience of notation we denote by
\[
n_{s_0} := \operatorname{coeff}^{D_{s_0}}\left(\prod\limits_{k=1}^N \deg_{\mathcal V_{s_k}}\right),
\]
is specified by the formula
\begin{align*} 
n_{s_0} &= \sum\limits_{\substack{ I \subset \{s_1,\ldots,s_N\} \\ I \neq \emptyset}} (2)^{\vert I \vert -1} (-1)^{|I|} [s_0 = \operatorname{gcd}(I) ] \\
&=  - [s_0 \in \{ s_1,\ldots,s_N \} ] + 2 \sum\limits_{\substack{ I \subset \{s_1,\ldots,s_N\} \\ I \neq \emptyset, \{s_0\} }} (-2)^{\vert I \vert-2} 
 \left[s_0 = \operatorname{gcd}(I)  \right].
\end{align*}
\end{proof}

\section{The Burnside Ring Cryptosystem (BRC)}

We now formalize the cryptosystem, which operates entirely within the full Burnside ring $A(G)$ of a compact Lie group $G$. The protocol relies on a shared secret key that defines an involutory transformation on the ring.

\subsection{The Shared Secret Key}
Before communication can begin, the two parties (Alice and Bob) must privately agree on a secret key. In the BRC scheme, the key is a tuple $(G, \mathcal{O}, S)$ whose components determine the algebraic environment and the specific transformation used for encryption.

\begin{enumerate}
    \item \textbf{The Group ($G$):} The choice of a compact Lie group, for example $G=O(2)$, defines the entire algebraic universe in which the cryptography operates. The group's structure dictates the basis elements of the Burnside ring $A(G)$ and, most importantly, the rules for multiplication.
    
    \item \textbf{The Basis Ordering ($\mathcal{O}$):} This is a chosen total ordering of the basis elements in $\Phi_0(G)$. This ordering acts as a secret permutation that determines how a standard vector of integers is mapped to a specific element in the Burnside ring.
    
    \item \textbf{The Representation Indices ($S$):} This is a finite set of indices, $S = \{k_1, k_2, \dots, k_m\}$, corresponding to a selection of irreducible representations of $G$. This compact set of integers is the seed for generating the key transformation.
\end{enumerate}

From this secret key, both parties deterministically construct the key element, $k \in A(G)$, which is the product of the \emph{basic degrees} associated with the representations in $S$. For our working example $G=O(2)$, this key is $k = \prod_{j \in S} \deg_{\mathcal{V}_j}$. As established in the previous section, each basic degree is an involution, and since they commute, their product $k$ is also an involution, satisfying $k^{-1} = k$.

\subsection{The Encryption and Decryption Protocol}
The protocol allows for the encryption of any data that can be represented as a finite sequence of integers. The process is illustrated in Figure \ref{fig:flow}.

\begin{figure}
    \centering
\includegraphics[width=\linewidth]{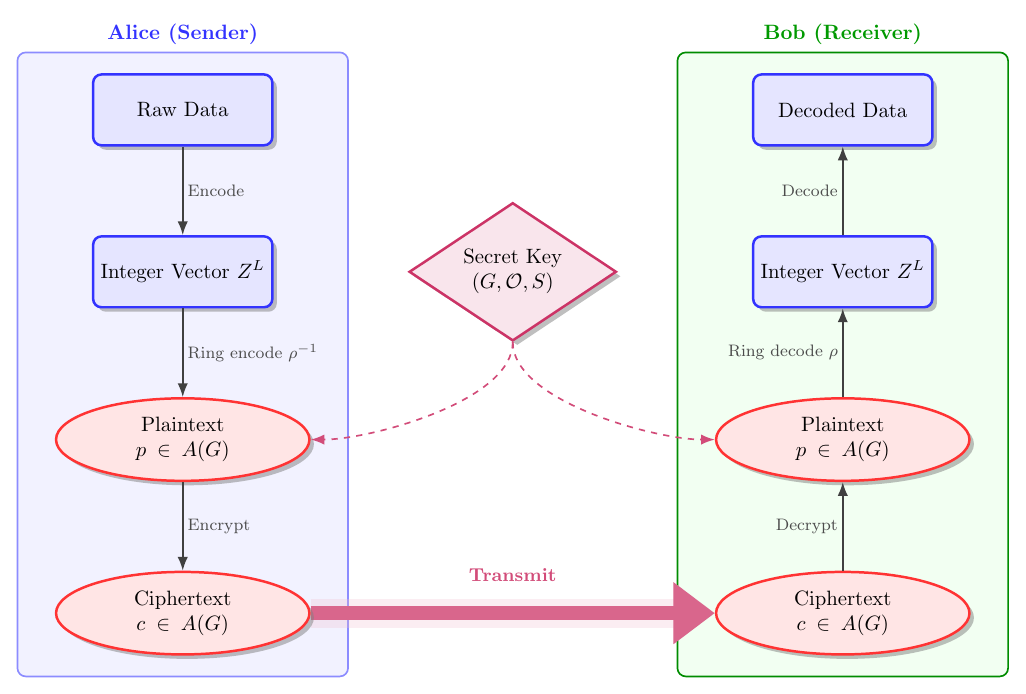}
    \caption{The flow of the Burnside Ring Cryptosystem.}
    \label{fig:flow}
\end{figure}
Let us assume Alice wishes to send a secret message to Bob.

\begin{enumerate}
    \item \textbf{Preprocessing (Alice):} Alice takes her raw data (e.g., a text string) and converts it into a vector of integers, $\vec{p} \in \bz^L$, using a public, standard encoding scheme like ASCII. The length $L$ of the vector is determined by the length of her message.

    \item \textbf{Ring Encoding (Alice):} Alice uses the secret basis ordering $\mathcal{O}$ to map her vector $\vec{p}$ to a plaintext element $p \in A(G)$. For instance, if $\mathcal{O}$ enumerates the basis as $((H_1), (H_2), \dots)$, then $p = \sum_{i=1}^L p_i (H_i)$.

    \item \textbf{Encryption (Alice):} Alice computes the ciphertext element $c = p \cdot k$ using the multiplication rules of the shared group's Burnside ring, $A(G)$.

    \item \textbf{Transmission:} Alice transmits the ciphertext. Since $c$ is an element with finite support, this involves sending the finite set of its non-zero coefficients and their corresponding basis indices.

    \item \textbf{Decryption (Bob):} Bob receives the transmitted data and reconstructs the ciphertext element $c$. He then applies the exact same transformation using his identical, secretly-held key: $p = c \cdot k$. Because the key is an involution, this operation reverses the encryption, recovering the original plaintext element.

    \item \textbf{Decoding (Bob):} Bob applies the coefficient map $\rho$ using the secret ordering $\mathcal{O}$ to map the element $p$ back to the integer vector $\vec{p}$. He then uses the public decoding scheme to retrieve Alice's original raw data.
\end{enumerate}

\subsection{The $O(2)$ Burnside Ring Cryptosystem}

For the remainder of this paper, we focus on a concrete implementation of the BRC using the group $G=O(2)$. To simplify the framework and focus on the core algebraic properties, we make the following assumptions:
\begin{itemize}
    \item The group $G=O(2)$ is fixed and public.
    \item The secret basis ordering $\mathcal{O}$ is simplified to a public, natural ordering where the $i$-th component of a message vector corresponds to the basis element $(D_i)$.
    \item The message space is restricted to the submodule spanned by the dihedral subgroups. A plaintext vector $\vec{p} = (p_1, p_2, \dots, p_L)$ is encoded as the element $p = \sum_{i=1}^L p_i (D_i) \in A(O(2))$.
\end{itemize}
\begin{figure}[h!]
    \centering
\includegraphics[width=\textwidth]{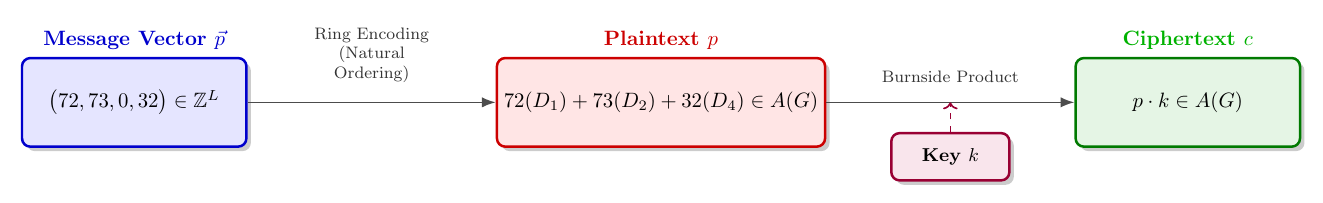}
    \caption{An illustration of the encoding and encryption process for the $O(2)$-BRC under the simplifying assumptions.}
    \label{fig:o2_protocol}
\end{figure}
Under these conditions, the only remaining secret component of the key is the set of representation indices $S$. The following proposition establishes a crucial property of this specialized system.

\begin{proposition}[Support Preservation] \label{prop:support}
Let the plaintext be an element $p = \sum_{i=1}^L a_i (D_i)$. If the key $k$ is any product of basic degrees, the ciphertext $c = p \cdot k$ is also an element of the submodule $\mathbb{Z}[\{(D_1), \dots, (D_L)\}]$. Encryption does not expand the support of the message to higher-order dihedral basis elements.
\end{proposition}
\begin{proof}
The key $k$ is a linear combination of $(O(2))$ and various $(D_m)$ basis elements. We consider the product of a plaintext basis element $(D_i)$ (with $i \le L$) with the components of $k$:
\begin{itemize}
    \item $(D_i) \cdot (O(2)) = (D_i)$.
    \item $(D_i) \cdot (D_m) = 2(D_{\gcd(i,m)})$.
\end{itemize}
Since $\gcd(i,m) \le i \le L$, any resulting dihedral component in the ciphertext will have an index no larger than $L$. Thus, the support of $c$ is contained within the submodule spanned by the original plaintext's basis.
\end{proof}

\section{Security Analysis of the O(2)-BRC}
\label{sec:security}
We analyze the $O(2)$-BRC as defined above, where plaintexts are restricted to the dihedral submodule. We consider a known-plaintext attack model, assuming that the adversary has observed a finite collection of plaintext-ciphertext pairs $\{(p_j, c_j)\}_{j=1}^N$, where $c_j = p_j \cdot k$. Although the cryptosystem is linear, we will demonstrate that its operation in an infinite-dimensional space prevents key recovery by simply solving a finite system of linear equations.

\subsection{The Finite Observation Window}
Any finite set of observed plaintexts $\{p_j\}$ involves a finite number of basis elements. Let $L$ be the maximal dihedral index appearing in any of the plaintexts. By Proposition~\ref{prop:support}, the corresponding ciphertexts $\{c_j\}$ will also have supports within this range. Therefore, all observations are confined to the finite-rank $\mathbb{Z}$-module
\[
W_L := \mathbb{Z}[\{ (D_1),\dots,(D_L) \}] \subset A(O(2)).
\]
An adversary with enough linearly independent plaintexts can determine the operator $T_k(p)=p \cdot k$ restricted to $W_L$, but recovering the operator is not the same as recovering the unique key set $S$.

\subsection{Key Indistinguishability from Finite Data}
The core of the passive security of the $O(2)$-BRC lies in the fact that a countably infinite number of distinct key sets $S$ can generate operators that are identical on $W_L$.

\begin{proposition}[Key Ambiguity on $W_L$]
\label{prop:indistinguishability}
Let $S = \{s_1, \dots, s_m\}$ be a secret key set. Let $q$ be any prime number such that $q > L$. Construct a new set $S' = \{s'_1, \dots, s'_m\}$ by defining $s'_j = s_j q$. Then the key elements $k_S$ and $k_{S'}$ generate identical linear transformations on $W_L$.
\end{proposition}

\begin{proof}
The linear transformation on $W_L$ is determined by its action on the basis elements $(D_i)$ for $1 \le i \le L$. The effect of the key $k_S$ on $(D_i)$ is determined by terms involving $(D_{\gcd(i, \gcd(I))})$ for subsets $I \subseteq S$. For the key $k_{S'}$, the relevant term involves $(D_{\gcd(i, \gcd(I'))})$, where $I' = \{s q : s \in I\}$. By the properties of the greatest common divisor, $\gcd(I') = q \gcd(I)$. To show the transformations are identical, we must demonstrate that for any $i \le L$ and any $I \subseteq S$, one has
\[
\gcd(i, q \cdot \gcd(I)) = \gcd(i, \gcd(I)).
\]
Clearly, this equality holds if $\gcd(i, q) = 1$. By assumption, $q$ is a prime number and $q > L$. Since we are considering a basis element $(D_i)$ where $i \le L$, it follows that $i < q$. Because $q$ is a prime larger than $i$, it cannot be a prime factor of $i$. 
\end{proof}

\begin{corollary}[Infinite Non-Identifiability from Finite Data]
\label{cor:infinite-nonid}
For any key $S$ consistent with a finite set of observations in $W_L$, there exist infinitely many distinct key sets $S'$ that are also consistent with all observations. The key set $S$ is therefore information-theoretically non-identifiable from any finite data.
\end{corollary}

\begin{proof}
One can construct an infinite family of such sets $S'$ by taking distinct primes $q_1, q_2, \dots > L$ and defining $S'_j = \{s \cdot q_j \mid s \in S\}$. Each $S'_j$ is distinct from $S$ and from each other, but by Proposition~\ref{prop:indistinguishability}, all generate an operator indistinguishable from $T_{k_S}$ on $W_L$.
\end{proof}

\subsection{Vulnerability to Chosen-Plaintext Attack (CPA)}
\label{subsec:CPA}
While the BRC resists key recovery in passive models (ciphertext-only and known-plaintext) due to the infinite-dimensional ambient space and the resulting non-identifiability from finite data, it is not designed to withstand \emph{chosen-plaintext} attacks. In this subsection we formalize a key-distinguishing experiment and give a deterministic, information-theoretic winning strategy for the adversary that uses only queries of the form $p=(D_x)$.

\medskip
\noindent\textbf{Key form and coefficients.}
For a finite key index set $S=\{s_1,\dots,s_N\}$ with distinct positive integers and
\[
k_S = \prod_{j=1}^N \deg_{\mathcal V_{s_j}} = (O(2)) + \sum_{n\ge 1} \alpha_S(n)\,(D_n),
\]
where, by Proposition~\ref{prop:key-coeff}, the dihedral coefficients $\alpha_S(s_0) \in \bz$ are given by
\begin{equation}\label{eq:key-coeff}
\alpha_S(s_0)
=
-\,[s_0\in S]
+
2\!\!\!\sum_{\substack{I\subset S\\ I\neq \emptyset,\ \{s_0\}}}
(-2)^{|I|-2}\,[\,s_0=\gcd(I)\,],
\qquad s_0\in\bn.
\end{equation}
\vs
\noindent\textbf{Oracle responses to dihedral probes.}
For any $x\in\bn$ and any key $S$, the Burnside product rules  imply
\begin{equation}\label{eq:oracle}
(D_x)\cdot k_S
=
(D_x) + \sum_{n\ge 1}2\alpha_S(n) \,(D_{\gcd(x,n)}).
\end{equation}
In particular, the coefficient of $(D_x)$ in the ciphertext $c_x := (D_x) \cdot k_S$ equals
\begin{equation}\label{eq:Dx-coeff}
\operatorname{coeff}^{D_x}\big((D_x)\!\cdot\! k_S\big)
=
1 + 2\sum_{\substack{n\ge 1\\ x\,\mid\, n}} \alpha_S(n).
\end{equation}
Defining the arithmetic function
\begin{align}\label{def:counting_function}
\beta_S(x)\;:=\;\sum_{\substack{n\ge 1\\ x\,\mid\, n}} \alpha_S(n)
\quad\text{for }x\in\bn,
\end{align}
then \eqref{eq:Dx-coeff} reads $\operatorname{coeff}^{D_x}((D_x)\cdot k_S)=1+2\,\beta_S(x)$.

\medskip
\noindent\textbf{CPA key-distinguishing experiment.}
The adversary is given two distinct candidate key sets $S_0\neq S_1$ and access to an oracle that encrypts under $k_{S_b}$ for a hidden $b\in\{0,1\}$:
\begin{enumerate}
  \item[\emph{(1)}] The adversary may adaptively query plaintexts $p\in A(O(2))$ (we will demonstrate that a clever adversary only needs queries of the form $p=(D_x)$).
  \item[\emph{(2)}] For each query, the oracle returns the corresponding cipher text $c=p\cdot k_{S_b}$.
  \item[\emph{(3)}] After a finite number of queries, the adversary must output a guess $\hat b$; success means that
  the probability of $\hat b=b$ is one.
\end{enumerate}

\begin{proposition}[CPA Distinguishability by Dihedral Probes]
\label{prop:CPA-distinguish}
For any two distinct key sets $S_0\neq S_1$, there exists $x\in\bn$ such that
\[
(D_x)\cdot k_{S_0} \;\neq\; (D_x)\cdot k_{S_1}.
\]
Equivalently, there exists $x$ with $\beta_{S_0}(x)\neq \beta_{S_1}(x)$. Consequently, the adversary can win the key-distinguishing CPA experiment with a single query $p=(D_x)$.
\end{proposition}

\begin{proof}
Since $S_0 \neq S_1$, their symmetric difference $S_0 \triangle S_1$ is non-empty. Let the adversary choose $x^* = \max(S_0 \triangle S_1)$, the largest element in this set, and assume without loss of generality that $x^* \in S_0 \setminus S_1$. By Equation~\eqref{eq:key-coeff}, the parity of the coefficient $\alpha_S(n)$ is determined entirely by the term $-[n \in S]$. Thus, $\alpha_{S_0}(x^*)$ is odd while $\alpha_{S_1}(x^*)$ is even.

Adopting the notation \eqref{def:counting_function}, consider the difference 
\[
\beta_{S_0}(x^*) - \beta_{S_1}(x^*) = (\alpha_{S_0}(x^*) - \alpha_{S_1}(x^*)) + \sum_{\substack{n>x^* \\ x^*|n}} (\alpha_{S_0}(n) - \alpha_{S_1}(n)).
\]

For any $n$ in the summation, $n > x^*$, which implies $n \notin S_0 \triangle S_1$. Therefore, for each such $n$, either $n$ is in both $S_0$ and $S_1$, or in neither. In both cases, $\alpha_{S_0}(n)$ and $\alpha_{S_1}(n)$ have the same parity, making their difference an even integer. It follows that the coefficient of $(D_{x^*})$ in the query difference $(k_{S_0} - k_{S_1}) \cdot (D_{x^*})$ is odd and, in particular, non-zero. 
\end{proof}
\vs
\noindent\textbf{A one-query distinguisher.}
Given $S_0,S_1$, the adversary computes $\beta_{S_0}(x)$ and $\beta_{S_1}(x)$ from \eqref{eq:key-coeff} and selects any $x$ with $\beta_{S_0}(x)\neq \beta_{S_1}(x)$ (existence is guaranteed by Proposition~\ref{prop:CPA-distinguish}). In particular, if the adversary selects the largest index $x^*$ at which $S_0$ and $S_1$ disagree then the single query $p=(D_{x^*})$ returns $c=(D_x)\cdot k_{S_b}$; reading $\operatorname{coeff}^{D_{x^*}}(c)=1+2\,\beta_{S_b}(x^*)$, determining $b$ with certainty.

\begin{remark}[Trivial distinguisher via $(O(2))$]
\label{rem:O2-query}
Since $(O(2))$ is the multiplicative identity in $A(O(2))$, the query $p=(O(2))$ returns $c=k_{S_b}$ directly, trivially revealing $b$ in one query. If one wishes to exclude this triviality, one may restrict admissible plaintexts to the dihedral span $\langle (D_n):n\in\bn\rangle_{\bz}$; Proposition~\ref{prop:CPA-distinguish} then shows that even under this restriction a single dihedral probe suffices to distinguish any two distinct keys.
\end{remark}

It follows that the BRC, at least with $G = O(2)$, is \emph{not} IND-CPA secure: an active adversary can always distinguish between two candidate keys with probability $1$ using at most one carefully chosen dihedral query. This limitation is inherent to deterministic linear schemes and does not affect the non-identifiability guarantees we establish for passive adversarial attack models.
\appendix
\section{The G-Equivariant Brouwer Degree}\label{sec:appendix_eqdeg}
Let $V$ be an orthogonal $G$-representation with an open bounded $G$-invariant set $\Om \subset V$. A $G$-equivariant map $f:V \rightarrow V$ is said to be $\Om$-admissible if $f(x) \neq 0$ for all $x \in \partial \Om$, in which case the pair $(f,\Om)$ is called an admissible $G$-pair in $V$. We denote by $\mathcal M^{G}(V)$ the set of all admissible $G$-pairs in $V$ and by $\mathcal{M}^{G}$ the set of all admissible $G$-pairs defined by taking a union over all orthogonal $G$-representations, i.e.
\[
\mathcal M^{G} := \bigcup\limits_V \mathcal M^{G}(V).
\]
The $G$-equivariant Brouwer degree provides an algebraic count of solutions, according to their symmetric properties, to equations of the form
\[
f(x) = 0, \; x \in \Omega,
\]
where $(f, \Omega) \in \mathcal M^{G}$. In fact, it is standard (cf. \cite{book-new}, \cite{AED}) to define the {\it $G$-equivariant Brouwer degree} as the unique map associating to every admissible $G$-pair $(f,\Om)\in \mathcal M^{G}$ an element from the Burnside ring $A(G)$, satisfying the four {\it degree axioms} of additivity, homotopy and normalization:
\vs
\begin{theorem} \rm
\label{thm:GpropDeg} There exists a unique map $\gdeg:\mathcal{M}
	^{G}\to A(G)$, that assigns to every admissible $G$-pair $(f,\Omega)$ the Burnside ring element
	\begin{equation}
		\label{eq:G-deg0}\gdeg(f,\Omega)=\sum_{(H) \in \Phi_0(G)}%
		{n_{H}(H)},
	\end{equation}
	satisfying the following properties:
    \begin{enumerate}[label=(G$_\arabic*$)]
		\item\label{g2} \textbf{(Additivity)} 
  For any two  disjoint open $G$-invariant subsets
  $\Omega_{1}$ and $\Omega_{2}$ with
		$f^{-1}(0)\cap\Omega\subset\Omega_{1}\cup\Omega_{2}$, one has
		\begin{align*}
\gdeg(f,\Omega)=\gdeg(f,\Omega_{1})+\gdeg(f,\Omega_{2}).
		\end{align*}
		\item\label{g3} \textbf{(Homotopy)} For any 
  $\Omega$-admissible $G$-homotopy, $h:[0,1]\times V\to V$, one has
		\begin{align*}
\gdeg(h_{t},\Omega)=\mathrm{constant}.	\end{align*}
		\item\label{g4}  \textbf{(Normalization)}
  For any open bounded neighborhood of the origin in an orthogonal $G$-representation $V$ with the identity operator $\id:V \rightarrow V$, one has
		\begin{align*}
	\gdeg(\id,\Omega)=(G).
		\end{align*}
	\end{enumerate}
\end{theorem}
The following are additional properties of the map $\gdeg$ which can be derived from the four axiomatic properties defined above (cf. \cite{book-new}, \cite{AED}):		
\begin{enumerate}[label=(G$_\arabic*$), start=4]
\item\label{g4} \textbf{(Existence)} If  $n_{H} \neq0$ for some $(H) \in \Phi_0(G)$ in \eqref{eq:G-deg0}, then there
		exists $x\in\Omega$ such that $f(x)=0$ and $(G_{x})\geq(H)$.
		\item\label{g5}  {\textbf{(Multiplicativity)}} For any $(f_{1},\Omega
		_{1}),(f_{2},\Omega_{2})\in\mathcal{M} ^{G}$,
		\begin{align*}
			\gdeg(f_{1}\times f_{2},\Omega_{1}\times\Omega_{2})=
		\gdeg(f_{1},\Omega_{1})\cdot \gdeg(f_{2},\Omega_{2}),
		\end{align*}
		where the multiplication `$\cdot$' is taken in the Burnside ring $A(G)$.
\item\label{g6} \textbf{(Recurrence Formula)} For an admissible $G$-pair
		$(f,\Omega)$, the $G$-degree \eqref{eq:G-deg0} can be computed using the
		following Recurrence Formula:
		\begin{equation}
	\label{eq:RF-0}n_{H}=\frac{\deg(f^{H},\Omega^{H})- \sum_{(K)\succ(H)}{n_{K}\,
					n(H,K)\, \left|  W(K)\right|  }}{\left|  W(H)\right|  },
		\end{equation}
		where $\left|  X\right|  $ stands for the number of elements in the set $X$
		and $\deg(f^{H},\Omega^{H})$ is the Brouwer degree of the map $f^{H}%
		:=f|_{V^{H}}$ on the set $\Omega^{H}\subset V^{H}$.
	\end{enumerate}
    
\subsection{Basic Degrees and a Formula for the $G$-Equivariant Brouwer Degree of a $G$-equivariant Linear Isomorphism}
Let $T: V \rightarrow V$ be a $G$-equivariant linear isomorphism on an orthogonal $G$-representation $V$.
Assuming that a complete list $\{ \mathcal{V}_j \}$ of the irreducible $G$-representations is made available, $V$ admits an {\bf $G$-isotypic decomposition} of the form
\begin{align} \label{eq:Gisotypic_decomp}
        V = \bigoplus_j V_j
\end{align}
where the $j$-th $G$-isotypic component $V_j$ is modeled on the corresponding $G$-irreducible representation $\mathcal V_j$. By Schur's Lemma, $T$ preserves \eqref{eq:Gisotypic_decomp} allowing us to write
\[
T = \bigoplus_{k} T_k, \quad T_k := T|_{V_k}:V_k \rightarrow V_k.
\]
The Multiplicativity property \ref{g5} of the $G$-equivariant degree, allows us to decompose the degree of $T$ on the unit ball $B(V)$ into a Burnside ring product over these components
\begin{align}\label{eq:degree_T}
\gdeg(T,B(V)) = \prod_{k} \gdeg(T_k, B(V_k)).
\end{align} 
Since $T_k$ is $G$-equivariant for all $k$, the generalized eigenspace $E(\beta)$ for any eigenvalue $\beta \in \sigma(T_k)$ constitutes a $G$-invariant subspace such that every isotypic component $V_k$ admits a spectral decomposition $V_k = \bigoplus_{\beta \in \sigma(T_k)} E(\beta)$ which is preserved by $T_k$. Applying multiplicativity again yields
\[
\gdeg(T_k, B(V_k)) = \prod_{\beta \in \sigma(T_k)} \gdeg(T_k|_{E(\beta)}, B(E(\beta))).
\]
Notice that the degree of the restriction $T_k|_{E(\beta)}$ depends on the sign of $\beta$. Indeed, if $\beta > 0$ or is complex, $T_k|_{E(\beta)}$ is homotopic to the identity, and its degree is the multiplicative identity in $A(G)$. Thus, only negative eigenvalues contribute non-trivially. For $\beta < 0$, $T_k|_{E(\beta)}$ is homotopic to $-\mathrm{Id}$ such that its degree is $\gdeg(-\mathrm{Id}, B(E(\beta)))$. Now, each eigenspace $E(\beta)$ is a sub-representation of $V_k$ and must therefore be a direct sum of some number of copies of the irreducible representation $\mathcal{V}_k$. Let this number be $\mu_k(\beta) := \dim E(\beta) / \dim \mathcal{V}_k$. Applying the multiplicativity property a third and final time, one has
\[
\gdeg(-\mathrm{Id}, B(E(\beta))) = \gdeg(-\mathrm{Id}, B(\mathcal{V}_k))^{\mu_k(\beta)}.
\]
To formalize the above construction, we put $\mu_k := \sum_{\beta \in \sigma_-(T)} \mu_k(\beta)$ and
associate to each irreducible $G$-representation $\mathcal V_k$ the corresponding {\bf $k$-th basic degree}
\[
\deg_{\mathcal V_k} := \gdeg(-\mathrm{Id}, B(\mathcal{V}_k)),
\]
such that 
\eqref{eq:degree_T} becomes
\begin{align*}
\gdeg(T,B(V)) = \prod_{\beta \in \sigma_-(T)} \prod_{k} (\deg_{\mathcal V_k})^{\mu_k(\beta)} = \prod_{k} (\deg_{\mathcal V_k})^{\mu_k}.
\end{align*} 
Notice that the coefficients $n_{H} := \operatorname{coeff}^H(\deg_{\mathcal V_k})$ specifying each of the basic degrees $\deg_{\mathcal V_k} \in A(G)$
can be practically computed, using the recurrence formula  \eqref{eq:RF-0}, as follows
\begin{align} \label{eq:RF-1}
n_{H}=\frac{(-1)^{\dim\mathcal{V} _{k}^{H}}- \sum_{(H)\prec(K)}{n_{K}\, n(H,K)\, \left|  W(K)\right|  }}{\left|  W(H)\right|  }.
\end{align}

\end{document}